\documentclass[a4paper,UKenglish,cleveref, autoref, thm-restate]{lipics-v2021}

\pdfoutput=1 %

\title{String Diagrammatic Trace Theory} %

\author{Matthew Earnshaw}{Department of Software Science, Tallinn University of Technology, Estonia \and \url{https://ioc.ee/~matt} }{matthew.earnshaw@taltech.ee}{https://orcid.org/0000-0001-8236-2811}{}%

\author{Pawe\l{} Soboci\'{n}ski}{Department of Software Science, Tallinn University of Technology, Estonia \and \url{https://ioc.ee/~pawel} }{pawel.sobocinski@taltech.ee}{https://orcid.org/0000-0002-7992-9685}{}

\authorrunning{M. Earnshaw and P. Soboci\'{n}ski} %

\Copyright{Matthew Earnshaw and Pawe\l{} Soboci\'{n}ski} %



\category{} %

\relatedversion{} %

\funding{This research was supported by the ESF funded Estonian IT Academy research measure (project 2014-2020.4.05.19- 0001) and Estonian Research Council grant PRG1210.}%

\acknowledgements{We would like to thank Chad Nester, Mario Román, and Niels Voorneveld for helpful conversations.}%

\nolinenumbers %

\hideLIPIcs

\EventEditors{John Q. Open and Joan R. Access}
\EventNoEds{2}
\EventLongTitle{42nd Conference on Very Important Topics (CVIT 2016)}
\EventShortTitle{CVIT 2016}
\EventAcronym{CVIT}
\EventYear{2016}
\EventDate{December 24--27, 2016}
\EventLocation{Little Whinging, United Kingdom}
\EventLogo{}
\SeriesVolume{42}
\ArticleNo{23}

\usepackage{float}
\usepackage{cleveref}
\usepackage{unicode-mario}
\usepackage{quiver}
\usepackage{tikz}
\usepackage{xcolor}
\usepackage[]{hyperref}
\definecolor{nordnight}{HTML}{4c566a}
\hypersetup{
  colorlinks = true,
  linkcolor = nordnight,
  citecolor = black,
   urlcolor = black,
}
\usepackage{mathrsfs}
\usepackage{stmaryrd}

\usetikzlibrary{shapes.geometric, arrows}
\usetikzlibrary{arrows,backgrounds,shapes.geometric,shapes.misc,matrix,arrows.meta,decorations.markings}

\pgfdeclarelayer{background}
\pgfdeclarelayer{edgelayer}
\pgfdeclarelayer{nodelayer}
\pgfsetlayers{background,edgelayer,nodelayer,main}
\tikzset{baseline=-0.5ex}
\definecolor{light-gray}{gray}{.7}
\tikzstyle{none}=[inner sep=0pt]
\tikzstyle{plain}=[inner sep=0pt]
\tikzstyle{black}=[circle, draw=black, fill=black, inner sep=0pt, minimum size=3.5pt]
\tikzstyle{black-faded}=[circle, draw=light-gray, fill=light-gray, inner sep=0pt, minimum size=4pt]
\tikzstyle{white}=[circle, draw=black, fill=white, inner sep=0pt, minimum size=3.5pt]
\tikzstyle{white-faded}=[circle, draw=light-gray, fill=white, inner sep=0pt, minimum size=4.5pt]
\tikzstyle{delay}=[fill=black, regular polygon, regular polygon sides=3,rotate=-90, scale=.55]
\tikzstyle{delay-op}=[fill=black, regular polygon, regular polygon sides=3,rotate=90, scale=.55]
\tikzstyle{reg}=[draw, fill=white, rounded rectangle, rounded rectangle left arc=none, minimum height=1em, minimum width=1em, node font={\scriptsize}]
\tikzstyle{coreg}=[draw, fill=white, rounded rectangle, rounded rectangle right arc=none, minimum height=1em, minimum width=1em, node font={\scriptsize}]
\tikzstyle{rn}=[circle, draw=red, fill=red, inner sep=0pt, minimum size=4pt]
\tikzstyle{place}=[circle, draw=black, fill=white, inner sep=0pt, minimum size=8pt]
\tikzstyle{box} = [rectangle, minimum width=1cm, minimum height=1cm,text centered, draw=black, thick, fill=none]
\tikzstyle{sbox} = [rectangle, minimum width=0.18cm, minimum height=0.18cm,text centered, draw=black, thick, fill=none]

\newcommand{\symDiag}{\raisebox{.25em}{%
\begin{tikzpicture}[scale=.3]
	\begin{pgfonlayer}{nodelayer}
		\node [style=none] (0) at (-0.75, -0.5) {};
		\node [style=none] (1) at (-0.75, 0.5) {};
		\node [style=none] (2) at (0.75, -0.5) {};
		\node [style=none] (3) at (0.75, 0.5) {};
	\end{pgfonlayer}
	\begin{pgfonlayer}{edgelayer}
		\draw [in=180, out=0] (1.center) to (2.center);
		\draw [in=180, out=0] (0.center) to (3.center);
	\end{pgfonlayer}
\end{tikzpicture}
}}
\newcommand{\comp}{\mathop{\fatsemi}}%

\newcommand{\genr}[1][\gamma]{\,\providecommand{\generatorName}{}\renewcommand{\generatorName}{#1}\raisebox{0em}{%
\begin{tikzpicture}[node distance=0.1cm]
  \draw (-0.4,-0.15) -- (-0.21,-0.15);
  \draw (-0.4,0.15) -- (-0.21,0.15);
  \node [style=none,scale=0.75] at (-0.35,0) {\Large $\svdots$};
  \node (gamma) [sbox] {$\generatorName$};
  \node [style=none,scale=0.75] at (0.35,0) {\Large $\svdots$};
  \draw (0.21,-0.15) -- (0.4,-0.15);
  \draw (0.21,0.15) -- (0.4,0.15);
\end{tikzpicture}
}\!\!}
\newcommand{\genrw}[1][\gamma]{\,\providecommand{\generatorName}{}\renewcommand{\generatorName}{#1}\raisebox{0em}{%
\begin{tikzpicture}[node distance=0.1cm]
  \draw (-0.5,-0.15) -- (-0.3,-0.15);
  \draw (-0.5,0.15) -- (-0.3,0.15);
  \node [style=none,scale=0.75] at (-0.45,0) {\Large $\svdots$};
  \node (gamma) [sbox] {$\generatorName$};
  \node [style=none,scale=0.75] at (0.45,0) {\Large $\svdots$};
  \draw (0.3,-0.15) -- (0.5,-0.15);
  \draw (0.3,0.15) -- (0.5,0.15);
\end{tikzpicture}
}\!\!}

\usetikzlibrary{calc}

\tikzset{
  wire/.style={rounded corners=3, line width=.5pt}
, iden/.style={line width=.5pt}
, dot/.style={draw,fill=black, circle, inner sep=1pt}
, square/.style={draw,fill=white, rectangle, inner sep=1.5pt}
, wdot/.style={draw,fill=white, circle, inner sep=1pt}
, edot/.style={draw,fill=red!40, circle, inner sep=1pt}
, cdot/.style={draw,fill=#1, circle, inner sep=1pt}
, labeled/.style={draw,fill=white, inner sep=1pt}
}

\newcommand*{\Scale}[2][4]{\scalebox{#1}{\ensuremath{#2}}}%

\makeatletter
\newcommand{\svdots}{%
  \vbox{\fontsize{\sf@size}{\sf@size pt}\linespread{0.3}\selectfont
    \kern0.2\baselineskip
    \hbox{.}\hbox{.}\hbox{.}%
    \kern0.1\baselineskip
  }%
}
\makeatother
\makeatletter
\newcommand{\nicelinktarget}[1]{\Hy@raisedlink{\Hy@raisedlink{\hypertarget{#1}{}}}}
\makeatother
\newcommand\defining[2]{\nicelinktarget{#1}{\color{black}{#2}}}

\usepackage{xparse}

\NewDocumentCommand\Par{o}
{
  \IfNoValueTF{#1}
  {\textsf{Par}}
  {\textsf{Par}_{#1}}
}

\NewDocumentCommand\CRel{o}
{
  \IfNoValueTF{#1}
  {\textsf{CRel}}
  {\textsf{CRel}_{#1}}
}

\NewDocumentCommand\Lang{o}
{
\IfNoValueTF{#1}
  {\mathscr{L}}
  {\mathscr{L}(#1)}
}

\NewDocumentCommand\Rel{oo}
{
  \IfNoValueTF{#1}
  {\textsf{Rel}}
  \IfNoValueTF{#2}
  {\textsf{Rel}_{#1}}
  {\textsf{Rel}_{#1,#2}}
}

\NewDocumentCommand\detr{o}
{
  \IfNoValueTF{#1}
  {\mathscr{D}}
  {\mathscr{D}(#1)}
}

\NewDocumentCommand\endo{O{\mathscr{C}}m}
{
  {#1}_{#2}
}

\NewDocumentCommand\loc{m}
{
  \textsf{loc}(#1)
}

\NewDocumentCommand\Fprop{}
{
  \mathscr{F}_{\!\Scale[0.4]{\symDiag}\!}
}

\NewDocumentCommand\Fprem{}
{
 \mathscr{F}_p
}

\NewDocumentCommand\Dist{}
{
  \mathsf{Dist}
}

\NewDocumentCommand\Ind{}
{
  \mathsf{Ind}
}

\NewDocumentCommand\traceEndo{m}
{
  \Fprop{#1}\!\left(\substack{1\\ \svdots \\ n}, \substack{1 \\ \svdots \\ n}\right)
}

\NewDocumentCommand\preEndo{m}
{
  \Fprem{#1}\!\left(\substack{1 \\ \svdots \\ n}, \substack{1 \\ \svdots \\ n}\right)
}

\NewDocumentCommand\stk{mm}
{
  \substack{#1 \\ \svdots \\ #2}
}

\newcommand\regularMonoidalGrammar{\hyperlink{linkregmongram}{regular monoidal grammar}}
\newcommand\regularMonoidalGrammars{\hyperlink{linkregmongram}{regular monoidal grammars}}

\newcommand\regularSymmetricMonoidalLanguage{\hyperlink{linkregmonlang}{regular symmetric monoidal language}}

\newcommand\regularSymmetricMonoidalLanguages{\hyperlink{linkregmonlang}{regular symmetric monoidal languages}}

\newcommand\symmetricMonoidalLanguage{\hyperlink{linkmonlang}{symmetric monoidal language}}
\newcommand\symmetricMonoidalLanguages{\hyperlink{linkmonlang}{symmetric monoidal languages}}

\newcommand\MonoidalGraphs{\hyperlink{linkmongraph}{Monoidal graphs}}
\newcommand\monoidalGraphs{\hyperlink{linkmongraph}{monoidal graphs}}
\newcommand\monoidalGraph{\hyperlink{linkmongraph}{monoidal graph}}

\newcommand\distribution{\hyperlink{linkdist}{distribution}}
\newcommand\DistributedAlphabets{\hyperlink{linkdist}{Distributed alphabets}}
\newcommand\distributedAlphabet{\hyperlink{linkdist}{distributed alphabet}}
\newcommand\distributions{\hyperlink{linkdist}{distributions}}

\newcommand\symmetricMonoidalAutomata{\hyperlink{linksymmonaut}{symmetric monoidal automata}}
\newcommand\symmetricMonoidalAutomaton{\hyperlink{linksymmonaut}{symmetric monoidal automaton}}

\newcommand\monoidalAutomata{\hyperlink{linkndmonaut}{monoidal automata}}
\newcommand\monoidalAutomaton{\hyperlink{linkndmonaut}{monoidal automaton}}

\newcommand\monoidalSemiAutomaton{\hyperlink{linkndsmonaut}{monoidal semi-automaton}}

\newcommand\MonoidalTraceLanguages{\hyperlink{linkmontracelang}{Monoidal trace languages}}
\newcommand\monoidalTraceLanguages{\hyperlink{linkmontracelang}{monoidal trace languages}}
\newcommand\monoidalTraceLanguage{\hyperlink{linkmontracelang}{monoidal trace language}}

\newcommand\traceLanguages{\hyperlink{linkmaztracelang}{trace languages}}
\newcommand\traceLanguage{\hyperlink{linkmaztracelang}{trace language}}
\newcommand\MazurkiewiczTraceLanguages{\hyperlink{linkmaztracelang}{Mazurkiewicz trace languages}}
\newcommand\MazurkiewiczTraces{\hyperlink{linkmaztracelang}{Mazurkiewicz traces}}

\newcommand\monoidalDistributedAlphabet{\hyperlink{linkmondist}{monoidal distributed alphabet}}
\newcommand\monoidalDistributedAlphabets{\hyperlink{linkmondist}{monoidal distributed alphabets}}

\newcommand\stringDiagrams{\hyperlink{linkstringdiag}{string diagrams}}

\newcommand\traceMonoid{\hyperlink{linktracemon}{trace monoid}}
\newcommand\traceMonoids{\hyperlink{linktracemon}{trace monoids}}

\newcommand\independenceRelation{\hyperlink{linkir}{independence relation}}
\newcommand\independenceRelations{\hyperlink{linkir}{independence relations}}

 \definecolor{mygreen}{RGB}{0, 158, 115}
\definecolor{mylightblue}{RGB}{86, 180, 233}
\begin{document}

\maketitle

\begin{abstract}
We extend the theory of formal languages in monoidal categories to the multi-sorted, symmetric case, and show how this theory permits a graphical treatment of topics in concurrency. In particular, we show that Mazurkiewicz trace languages are precisely \emph{symmetric monoidal languages} over \emph{monoidal distributed alphabets}. We introduce \emph{symmetric monoidal automata}, which define the class of regular symmetric monoidal languages. Furthermore, we prove that Zielonka's asynchronous automata coincide with symmetric monoidal automata over monoidal distributed alphabets. Finally, we apply the string diagrams for symmetric premonoidal categories to derive serializations of traces.
\end{abstract}

\section{Introduction}
\emph{Monoidal languages} \cite{earnshaw22} are a generalization of formal languages of words to formal languages of \emph{string diagrams}. String diagrams \cite{joyal91,Selinger2011} are a graphical representation of morphisms in \emph{monoidal categories}, introduced in \Cref{sec:props}. Monoidal categories can be considered \emph{2-dimensional monoids} \cite{BURRONI199343}: just as monoids are categories with one object, in which the morphisms are elements of the monoid, (strict) monoidal categories can also be defined as 2-categories with one object. Accordingly, \emph{monoidal languages} are subsets of morphisms in free monoidal categories, just as word languages are subsets of free monoids. \emph{Regular} monoidal languages are those specifiable by means of finitary grammars or automata. Our paper \cite{earnshaw22} introduced these devices and examined the properties of their languages in the case of single-sorted, planar monoidal categories. These include regular languages of words and trees, as well as languages of \emph{planar} string diagrams that are neither linear nor tree-like.

In this paper, motivated by concurrency theory, we extend this theory to \emph{coloured props}: multi-sorted monoidal categories with symmetries (\Cref{sec:props}). The resulting theory of \emph{symmetric} monoidal languages (\Cref{sec:languages}) captures languages of diagrams having multiple colours of string and in which strings may cross, permitting non-planar diagrams. In terms of concurrency, colours represent different \emph{runtimes}, or \emph{threads of execution}.%

Indeed, in \Cref{sec:mazurkiewicz} we show that Mazurkiewicz trace languages \cite{mazurkiewicz} are exactly symmetric monoidal languages over alphabets of a particular shape called \emph{monoidal distributed alphabets}. In \Cref{sec:automata} we introduce automata for symmetric monoidal languages, defining the class of \emph{regular} symmetric monoidal languages. Then, in \Cref{sec:asynch} we show that these are exactly the asynchronous automata of Zielonka \cite{zielonka} when instantiated over monoidal distributed alphabets. Finally, in \Cref{sec:serial} we use the algebra of symmetric \emph{premonoidal} categories to show how serialization of traces can be treated string-diagrammatically.

\subsection{Related work}
Our previous work \cite{earnshaw22} introduced monoidal languages in the planar, single-sorted case; that is, languages of morphisms in free \emph{pros}. Similar languages of graphs were studied by Bossut \cite{bossut}, but their underlying algebra was not made explicit. Here, we again leverage the algebraic perspective, extending our theory to symmetric multi-sorted monoidal categories (props).

In the introduction to Joyal \& Street's foundational work on string diagrams for monoidal categories \cite{joyal91}, it is suggested that string diagrams have a connection to the \emph{heaps} of Viennot \cite{viennot}. Heaps are known to be equivalent to Mazurkiewicz trace monoids (also known as partially commutative monoids) \cite{cartier}, but a precise formulation of the suggested relation with string diagrams has not appeared in the literature until now.

The notion of dependence graph \cite{dependence} has also been used to give a topological presentation of Mazurkiewicz traces. Our use of the algebra of monoidal categories, rather than graphs, has various advantages. For example, we can apply our language theory for monoidal categories to traces, and we see notions such as asynchronous automata arise naturally from this. It also suggests generalizations of trace languages, in particular going beyond the case of atomic actions (\Cref{rem:effectful}). Finally, it brings our work into proximity with the semantics of Petri nets and other formalisms for concurrency based on monoidal categories \cite{baez2021categories,lmcs:10825}.

\section{Monoidal Graphs, Props and their String Diagrams} \label{sec:props}

In this section we recall the basic definitions used in the following, including the specific flavour of monoidal categories known as \emph{props} \cite{maclane}, along with their string diagrams \cite{joyal91,Selinger2011}. Just as a category can be presented by a directed graph, (strict) monoidal categories can be presented by \emph{monoidal graphs}, a kind of multi-input, multi-output directed graph.

\begin{definition} \label{defn:mongraph} \defining{linkmongraph}
A monoidal graph $𝒢$ is a set $B_𝒢$ of boxes, a set $S_𝒢$ of sorts, and functions $s,t: B_𝒢 \rightrightarrows {S_𝒢}^{*}$ to the free monoid over $S_𝒢$, giving source and target boundaries of each box.
\end{definition}

The alphabets of monoidal languages will be finite \monoidalGraphs{}: those in which $B_{\mathcal{G}}$ and $S_{\mathcal{G}}$ are both finite sets. In fact, since we are interested in finite state machines over finite alphabets, we will work exclusively with finite \monoidalGraphs{}. Diagrammatically, a (finite) \monoidalGraph{} can be pictured as a collection of boxes, labelled by elements of $B_𝒢$ with strings entering on the left and exiting on the right, labelled by sorts given by the source and target functions. For example, the following depicts the \monoidalGraph{} $\mathcal{G}$ with $B_{\mathcal{G}} = \{\gamma, \gamma'\}, S_{\mathcal{G}} = \{A,B\}, s(\gamma) = AB, t(\gamma) = ABA, s(\gamma') = A, t(\gamma') = BB$:

\begin{figure}[H] \label{fig:monoidal-graph}
  \centering
  \includegraphics[scale=1.2]{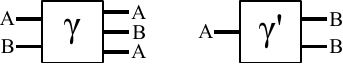}
\end{figure} \vspace*{-3mm}

Sorts of a \monoidalGraph{} are sometimes called \emph{colours}, since we could equally use different colours of string to represent different sorts, and we shall do so in places below. For a box $γ \in B_𝒢$ we call $s(γ)$ and $t(γ)$ the \emph{arity} and \emph{coarity} of $γ$, respectively, and write $γ : s(γ) → t(γ)$. We will also call $γ$ considered together with its arity and coarity a \emph{generator}.

\MonoidalGraphs{} are generating data for monoidal categories. Recall that a \emph{strict monoidal category} is a category $𝒞$, equipped with a functor $⊗ : 𝒞 × 𝒞 → 𝒞$ (the \emph{monoidal product}) and a unit object $I ∈ 𝒞$, such that $⊗$ is associative and unital. A strict monoidal category is \emph{symmetric} if there is a natural family of \emph{symmetry morphisms} $\sigma_{A,B} : A ⊗ B → B ⊗ A$, for each pair of sorts, satisfying $σ_{B,A} ∘ σ_{A,B} = 1_{A ⊗ B}$. The monoidal product turns the sets of objects and morphisms in $𝒞$ into monoids. A \emph{prop} is a symmetric strict monoidal category whose monoid of objects is a free monoid.\footnote{Some literature takes prop to mean that the monoid of objects is generated by a single object (and so isomorphic to $ℕ$), using the term \emph{coloured} prop for the general case above.}
While the above data can be intimidating to the non-expert, the free prop
$\Fprop{𝒢}$ on a \monoidalGraph{} $𝒢$ can be described in an intuitive and straightforward way: its arrows are the \emph{string diagrams} generated by $𝒢$.

\begin{definition} \label{defn:freeprop}
  The free prop $\Fprop{𝒢}$ on a \monoidalGraph{} $𝒢$ has monoid of objects $S_𝒢^{*}$ and morphisms \emph{string diagrams} inductively defined as follows\defining{linkstringdiag}{}:
  \begin{figure}[H]
    \centering
  \includegraphics[width=\textwidth]{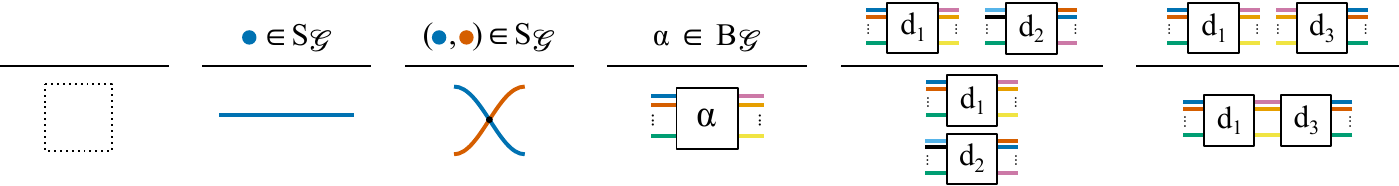}
  \end{figure} \vspace*{-2mm}
  Left to right: the empty diagram is a diagram; for every sort, the string on that sort is a diagram; for every pair of sorts, the symmetric braiding is a diagram; the diagram for every generator $α$ is a diagram; for any two diagrams their vertical juxtaposition is a diagram; and for any two diagrams with matching right and left boundaries, the diagram obtained by joining the matching wires is a diagram (their composition). The monoidal product is given on objects by concatenation, on diagrams by juxtaposition, and the unit is the empty word.
\end{definition}

The idea is simple: we treat generators like circuit components, and we have a supply of wires (identity morphisms). We also have the ability to cross wires, without tangling them; we do not distinguish over-crossings from under-crossings. A string diagram is then just any (open) circuit that we can build. This notation is sound and complete: an equation between morphisms of strict monoidal categories follows from their axioms if and only if it holds between string diagrams up to planar isotopy \cite{joyal91}. Working with string diagrams rather than the usual term syntax for morphisms is more intuitive, and leads to shorter proofs as the structural equations hold automatically: for example, interchange of morphisms (\Cref{fig:interchange}, left), unbraiding of symmetries (centre), and sliding of morphisms past symmetries (right).

\begin{figure}[H]
  \centering
  \includegraphics[width=\textwidth]{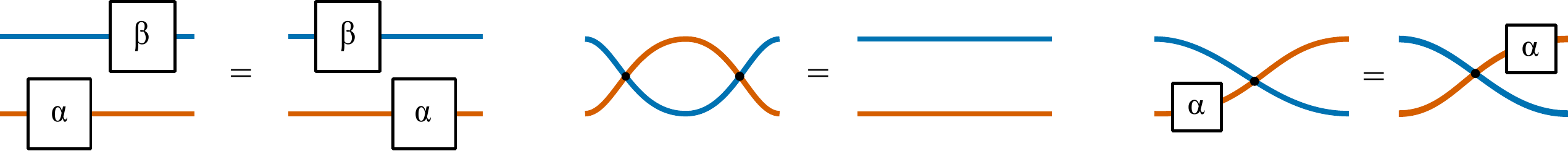}
  \caption{These pairs of string diagrams are equal, reflecting the functoriality of $⊗$ (\emph{interchange}), inverses of symmetries, and naturality of symmetries, respectively.} \label{fig:interchange}
\end{figure} \vspace*{-2mm}

\begin{definition} \label{defn:mgmorphism}
  A morphism of \monoidalGraphs{} $φ : ℋ → 𝒢$ is given by functions $B_φ : B_ℋ → B_𝒢$ and $S_φ : S_ℋ → S_𝒢$ compatible with source and target functions: $S_φ^{*} ∘ s = s ∘ B_φ$ and $S_φ^{*} ∘ t = t ∘ B_φ$, where $S_φ^{*}$ is the unique monoid homomorphism determined by $S_φ$.
\end{definition}

Morphisms of \monoidalGraphs{} freely generate morphisms of props: strict monoidal functors preserving colours. This extends to an adjunction $\Fprop{} \dashv \mathscr{U}$ between the categories of \monoidalGraphs{} and props, where $\mathscr{U}$ takes the underlying \monoidalGraph{} of a prop \cite{joyal91}. %

Monoidal categories have been applied to the study of both computing and physical processes \cite{COECKE201659,coecke2017picturing,dilavore2022coinductive,PAVLOVIC201394}. In these contexts, the monoidal product represents parallel composition of processes, and interchange reflects the \emph{independence} of processes running in parallel. This is the main feature of monoidal categories that we will leverage in our representation of \emph{traces} (\Cref{sec:mazurkiewicz}). The use of multi-sorted props will allow fine-grained control of interchange.

\section{Symmetric Monoidal Languages} \label{sec:languages}
 Our paper \cite{earnshaw22} treated the case of languages, grammars and automata over single-sorted \emph{pros} (strict monoidal categories \emph{without} symmetries), corresponding to languages of \emph{planar} string diagrams with one string colour. In this section we introduce the \emph{multi-sorted} (or ``coloured'') \emph{symmetric monoidal languages}, which will be needed in the following to extend monoidal language theory to trace theory. In \Cref{sec:automata} we introduce the corresponding automata.

 Just as a classical formal language is a subset of a free monoid, a symmetric monoidal language is a subset of morphisms in a free prop:

\begin{definition} \label{defn:symmonlang} \defining{linkmonlang}
  Let $Γ$ be a finite monoidal graph. A \emph{symmetric monoidal language} over $Γ$ is a set of morphisms in the free prop $\Fprop{Γ}$ over $Γ$.
\end{definition}

A morphism of finite directed graphs $G → Σ$, where $Σ$ is a graph with one vertex, amounts to a labelling of the edges of $G$ by edges of $Σ$. This is the starting point of Walters' definition of regular grammar \cite{WALTERS1989199}, which inspires the following definition:

\begin{definition} \defining{linkregmongram}
A regular monoidal grammar is a morphism of finite monoidal graphs. %
\end{definition}

For a \regularMonoidalGrammar{} $M  \overset{φ}{→} Γ$, the \monoidalGraph{} $Γ$ is the \emph{alphabet}, and the generators of $M$, with their labelling by $φ$, correspond to production rules: see \Cref{ex:circ}.

In the classical setting of word languages, a morphism of finite directed \emph{graphs} $G → Σ$ determines a \emph{regular language} over $Σ$ once we specify initial and final state vertices in $G$. In a \regularMonoidalGrammar{} $M → Γ$, the ``vertices'' of $M$ are words over $S_M$, leading to various natural choices of boundary condition (\Cref{rem:boundary}). In this paper, we will take initial and final words over $S_M^{*}$. Specifying these words defines the \symmetricMonoidalLanguage{} of the grammar (\Cref{defn:regmonlang}), and we define the languages arising in this way to be the \regularSymmetricMonoidalLanguages{}.

We illustrate these definitions with some pedagogical toy examples. In the remaining sections of this paper, we turn to our extended application in concurrency, and we shall see that Mazurkiewicz trace languages are a natural example of \symmetricMonoidalLanguages{}.

\begin{example}
  Let $φ : M → Γ$ be the \regularMonoidalGrammar{} where $M$ and $Γ$ have a single sort $(•)$ and no boxes, with $S_φ(•) = •$, and initial and final states $n \in S_M^{*} ≅ ℕ$. Then the \symmetricMonoidalLanguage{} of this grammar is the set of \emph{permutations} of $n$ wires: morphisms consisting only of symmetries and identities.
\end{example}

Props have been used to give syntax and semantics for various kinds of \emph{signal flow graph} and \emph{circuit diagrams} \cite{baez,boisseau,10.1145/2775051.2676993}. Intuitively, props are well suited for this purpose since wires may freely cross in a circuit.

\begin{example} \label{ex:circ}
We give a \regularMonoidalGrammar{} for the syntax of (open) circuits with $n \geqslant 0$ capacitors in series with a single voltage source (\Cref{fig:circuit}). The alphabet $Γ$ has a single sort, and boxes four circuit components (\Cref{fig:circuit}, left). The \monoidalGraph{} $M$ has four sorts $\{S,A,B,C\}$ and four boxes $s : S → AB, c : A → A, v : B → C, s' : AC → S$. $S_φ$ maps the four sorts to the single sort of $Γ$, and $B_φ$ maps each box to a circuit component. We can draw the grammar $φ : M → Γ$ in a single diagram by drawing the graph for $M$ but replacing each box $b$ with its image under the grammar morphism $B_φ(b)$ (\Cref{fig:circuit}, centre). The initial and final languages are the single state $\{S\}$. Intuitively, the \symmetricMonoidalLanguage{} determined by the grammar is all of the string diagrams $S → S$ that can be built using the ``sorted'' boxes of $Γ$, then forgetting the sorts.
\begin{figure}[h]
  \centering
  \includegraphics{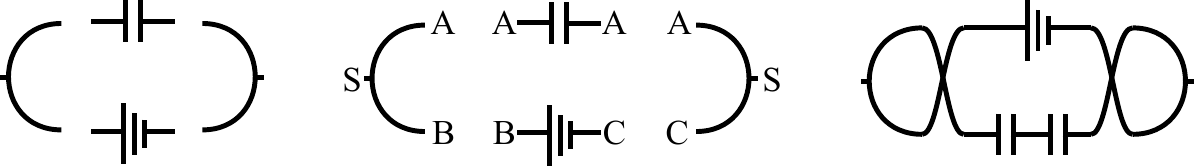}
  \caption{(Left) The alphabet $Γ$, giving syntax for circuits. (Centre) A regular monoidal grammar over $Γ$. (Right) An element of the regular symmetric monoidal language determined by this grammar.}
  \label{fig:circuit}
\end{figure}
\end{example}

\begin{remark} \label{rem:boundary}
As mentioned above, there are various possible choices for the ``initial and final states'' of a monoidal grammar. In our previous paper \cite{earnshaw22}, we took the empty word, giving languages of scalar string diagrams (i.e. no ``dangling wires''): this neatly generalizes tree grammars. More generally, one can take initial and final \emph{regular languages} of states over $S_M$, as considered by Bossut \cite{bossut}.
\end{remark}

The free prop construction can be used to concisely describe the \symmetricMonoidalLanguage{} of a \regularMonoidalGrammar{}, defining the class of \emph{regular symmetric monoidal languages}:

\begin{definition} \label{defn:regmonlang} \defining{linkregmonlang}
Let $(φ : M → Γ, I, F)$ be a regular monoidal grammar equipped with regular languages $I,F \subseteq S_M^{*}$. This determines a symmetric monoidal language by taking the image of the set of morphisms $\bigcup_{i \in I, f \in F}\Fprop{M}(i, f)$ under $\Fprop{φ}$, giving a set of morphisms in $\Fprop{Γ}$. The languages arising in this way are defined to be the \emph{regular} symmetric monoidal languages. %
\end{definition}

In this paper, we will only need the case where $I,F$ consist of single words. The slogan for the general case is that 2-dimensional regular languages have 1-dimensional regular boundaries. In \Cref{sec:automata}, we will see that \regularSymmetricMonoidalLanguages{} may equivalently be specified by \emph{non-deterministic monoidal automata}.

\begin{remark} \label{rem:gram}
A \regularMonoidalGrammar{} determines not only a \regularSymmetricMonoidalLanguage{}, but also a language in any algebraic structure generated by \monoidalGraphs{}, including planar monoidal categories (treated in \cite{earnshaw22}), and premonoidal categories (which we will use in \Cref{sec:serial}). This is analogous to the way in which a finite labelled directed graph may generate both a subset of a free monoid, but also a subset of a free group, by freely adding inverses to the graph. Moreover, many properties of planar regular monoidal languages such as their closure properties proved in \cite{earnshaw22} only use grammars, and hence the same proofs work for languages in these other algebras.
\end{remark}

\section{Mazurkiewicz Trace Languages as Symmetric Monoidal Languages} \label{sec:mazurkiewicz}
The theory of Mazurkiewicz traces \cite{booktraces,mazurkiewicz,mukund} provides a simple but powerful model of concurrent systems. Traces are a generalization of \emph{words} in which specified pairs of letters can commute. If we think of letters as corresponding to atomic \emph{actions}, then commuting letters reflect the \emph{independence} of those particular actions and so their possible concurrent execution: $ab$ is observationally indistinguishable from $ba$ if $a$ and $b$ are independent.

In this section, we show that trace languages are \symmetricMonoidalLanguages{} over \monoidalGraphs{} of a particular form that we call \emph{\monoidalDistributedAlphabets{}}. In Section \ref{sec:automata} we introduce \emph{symmetric \monoidalAutomata{}}, which operationally characterize the \regularSymmetricMonoidalLanguages{}. In Section \ref{sec:asynch} we turn to \emph{asynchronous automata} \cite{zielonka}, a well-known model accepting exactly the \emph{recognizable} trace languages, and show that these automata are precisely symmetric \monoidalAutomata{} over \monoidalDistributedAlphabets{}.

\subsection{Independence and distribution}

We recall some definitions from Mazurkiewicz trace theory, before recasting them in terms of monoidal languages. Fix a finite set $\Sigma$, an alphabet thought of as a set of atomic actions.

\begin{definition} \defining{linkir}
  An \emph{independence relation} on $\Sigma$ is a symmetric, irreflexive relation $I$. The induced \emph{dependence relation}, $D_I$ is the complement of $I$.
\end{definition}

\begin{definition} \label{defn:tracemonoid} \defining{linktracemon}
  For $I$ an \independenceRelation{}, let $\equiv_I$ be the least congruence on $Σ^{*}$ such that $∀ a,b$: $(a,b) \in I \implies ab \equiv_I ba$. The quotient monoid $\mathcal{T}(Σ,I) := \Sigma^{*}/{\equiv_I}$ is the \emph{trace monoid}.
\end{definition}

\begin{definition} \defining{linkmaztracelang}
  A (Mazurkiewicz) \emph{trace language} over $(Σ,I)$ is a subset of the \traceMonoid{} $\mathcal{T}(Σ,I)$. %
\end{definition}

An element of $\mathcal{T}(Σ,I)$ or \emph{trace over $(Σ,I)$} is thus an equivalence class of words up to commutation of independent letters. A trace language may be thought of as the set of possible observations of a concurrent system's behaviour, in which independent letters stand for actions which may occur concurrently. Independence relations correspond to \emph{distributions}:

\begin{definition}[\cite{mukund}] \label{defn:dist} \defining{linkdist}
  A \emph{distribution} of an alphabet $\Sigma$ is a finite tuple of non-empty alphabets $(\Sigma_1, ..., \Sigma_k)$ such that $\bigcup_{i=1}^k \Sigma_i = \Sigma$.
\end{definition}

\begin{proposition}[\cite{mukund}] \label{prop:loc}
A \distribution{} of $\Sigma$ corresponds to a function $\textsf{loc} : \Sigma \to \mathscr{P}^{+}(\{1,...,k\}) : σ ↦ \{ i \mid σ \in Σ_i \}$.
\end{proposition}

Such a function gives the set of ``locations'' of each action $\sigma \in \Sigma$. In terms of concurrency, we can consider this to be a set of memory locations, threads of execution, or runtimes in which $\sigma$ participates. In particular, every action has a non-empty set of locations.

A well-known construction \cite{mukund} allows us to move between \independenceRelations{} and \distributions{}: locations correspond to maximal cliques in the graph of the dependency relation. We recall this construction in the proof of Proposition \ref{prop:subposet}, which refines this correspondence.

Let $\Ind_\Sigma$ be the poset of \independenceRelations{} on $Σ$, with order the inclusion of relations. Similarly, define a preorder $\Dist_Σ$ on distributions by $(\Sigma_1, ..., \Sigma_p) \leqslant (\Sigma_1', ..., \Sigma_q')$ iff for each pair of distinct elements $a,b \in Σ$, if there exists $1 \leqslant j \leqslant q$ such that $Σ'_j$ contains both $a$ and $b$, then there exists an $Σ_i$ containing both $a$ and $b$. Finally, quotient this preorder by taking distributions to be equal up to permutation.

\begin{proposition} \label{prop:subposet}
 There is a Galois insertion $\Ind_Σ \hookrightarrow \Dist_Σ$.
\end{proposition}
\begin{proof}
  We construct an injective monotone function $i : \Ind_Σ \to \Dist_Σ$. Let an \independenceRelation{} $I$ over $Σ$ be given, with induced dependence relation $D_I$. Construct the undirected \emph{dependency graph}: vertices are elements of $\Sigma$ and there is an edge $(a,b)$ for every $(a,b) \in D_I$. Choose an ordering of maximal cliques of $D_I$, and define a \distributedAlphabet{} by taking $Σ_i$ to be the elements of $Σ$ in the maximal clique $i$. Different orderings give the same \distribution{} up to permutation, and so the same element of $\Dist_Σ$. This is injective since distinct \independenceRelations{} induce distinct dependency graphs. It is monotone since if $I \subseteq I'$ then the dependency graph $D_I$ is at least as connected as $D_{I'}$, so if $a,b$ both belong to a maximal clique of $D_{I'}$ then they will both belong to a maximal clique of $D_I$.

  We construct a monotone function $r : \Dist_Σ \to \Ind_Σ$. Let $(Σ_1, ..., Σ_k)$ be a \distribution{}. Define a relation $I$ by $(a, b) \in I \iff \textsf{loc}(a) \cap \textsf{loc}(b) = \varnothing$. This is irreflexive and symmetric, and so an \independenceRelation{}. $r$ is also clearly well-defined and monotone. Finally it is easy to check that $r ∘ i : \Ind_Σ \to \Ind_Σ$ is the identity. %

\end{proof}

Put otherwise, though the same \independenceRelation{} may be induced by many different \distributions{}, \independenceRelations{} correspond bijectively with the \distributions{} in the image of $i ∘ r$, that is, the \distributions{} obtained via the maximal clique construction.

\subsection{Symmetric monoidal languages over monoidal distributed alphabets}

We now turn to the interpretation of these notions in terms of \symmetricMonoidalLanguages{}. A \distribution{} can be seen as a \monoidalGraph{} in which sorts are the locations (runtimes). %

\begin{definition} \label{defn:mondist} \defining{linkmondist}
  A monoidal distributed alphabet is a finite \monoidalGraph{} $Γ$ with the following properties:
  \begin{itemize}
  \item $Γ$ has set of sorts a finite ordinal $S_Γ = \{1 < 2 < ... < k\}$ for $k \geqslant 1$,
  \item sorts $i \in S_Γ$ appear in order in the sources and targets of each generator $γ \in B_Γ$,
  \item each sort $i \in S_Γ$ appears at most once in each source and target,
  \item for each generator $γ \in B_Γ$, the sources and targets are non-empty and equal: $s(γ) = t(γ)$. %
  \end{itemize}
\end{definition}

In brief, every generator in the alphabet is equipped with some set of runtimes, which serve as its source and target, and the runtimes are conserved. \Cref{fig:dist-example} gives an example.
\begin{figure}[H]
  \centering
  \includegraphics[scale=0.6]{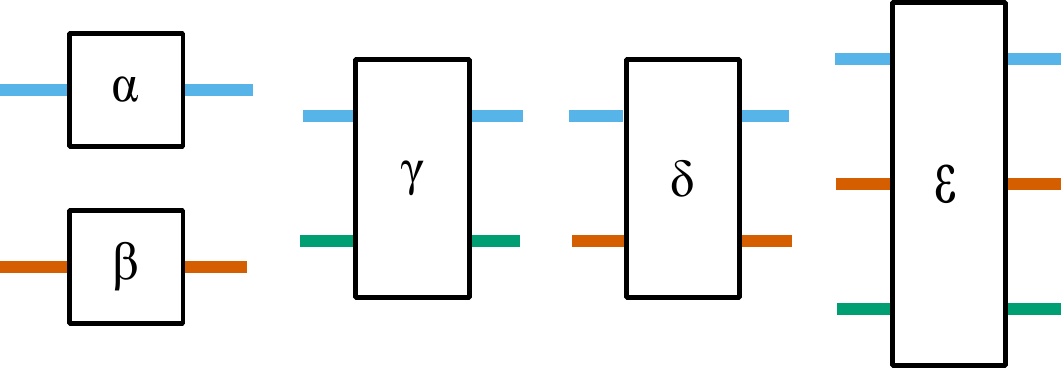}
  \caption{An example of a \monoidalDistributedAlphabet{}. For example, $δ$ and $\beta$ are independent but $γ$ and $α$ are not. We use colours for clarity, here blue = 1 < red = 2 < green = 3.} \label{fig:dist-example}
\end{figure} \vspace*{-2mm}

This gives us a way of representing \distributions{} as \monoidalGraphs{} and vice-versa, if the graph is a \monoidalDistributedAlphabet{}. Following \Cref{prop:loc}, we will use $\textsf{loc}(\genr{})$ to mean the arity (= coarity) of a generator $\genr{}$. Since we choose a finite ordinal for the sorts, we have that:

\begin{proposition} \label{prop:dists}
  \DistributedAlphabets{} are in bijection with \monoidalDistributedAlphabets{}.
\end{proposition}

Since the \emph{ordering} of the runtimes is ultimately not relevant to the structure of a trace, we should allow them to freely cross each other in our string diagrams: this is precisely what is enabled by taking the \emph{symmetric} monoidal languages over these alphabets. We also need each runtime to appear once in each element of these languages, so we take the boundaries to be $1 ⊗ ... ⊗ n$, which we will write as $\substack{1 \\ \svdots \\ n}$.

\begin{definition} \defining{linkmontracelang}
A \emph{monoidal trace language} is a \symmetricMonoidalLanguage{} of the form $L \subseteq  \traceEndo{Γ}$ where $Γ$ is a \monoidalDistributedAlphabet{}.
\end{definition}

\Cref{fig:trace-example} gives an example of an element in a \monoidalTraceLanguage{} over the \monoidalDistributedAlphabet{} in \Cref{fig:dist-example}. We call such morphisms \emph{monoidal traces}, and indeed we shall see below that they are exactly \MazurkiewiczTraces{}. The corresponding string diagram gives an intuitive representation of traces as topological objects.%

\begin{figure}[H]
  \centering
  \includegraphics[scale=0.6]{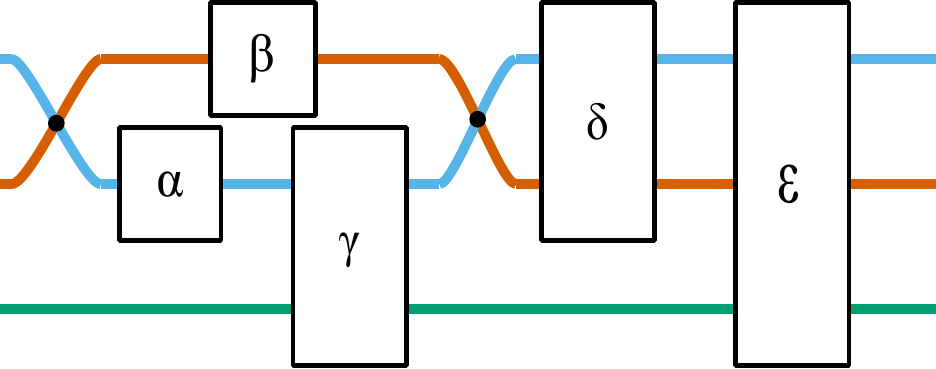}
  \caption{An example of a monoidal trace. $\beta$ is independent of $\alpha$ and $\gamma$, but not $\delta$ or $\epsilon$. Thus $αγβδε$ and $βαγδε$ are two possible serializations of this trace, corresponding to sliding $β$ past $α$ and $γ$ in the string diagram. We use colours for sorts, blue = 1 < red = 2 < green = 3.} \label{fig:trace-example} %
\end{figure} \vspace*{-2mm}

We now show that \monoidalTraceLanguages{} correspond precisely to \MazurkiewiczTraceLanguages{} (\Cref{thm:monoidaltrace}), by establishing an isomorphism of monoids between \traceMonoids{} and monoids of \stringDiagrams{} generated by \monoidalDistributedAlphabets{}. Fix a \monoidalDistributedAlphabet{} $Γ$. Recall that endomorphism hom-sets in a category are monoids under composition, and that the hom-set $\traceEndo{Γ}$ has elements string diagrams $\substack{1 \\ \svdots \\ n} → \substack{1 \\ \svdots \\ n}$ over $Γ$.

\begin{lemma} \label{prop:endo-pres}
  The hom-set $\traceEndo{Γ}$ admits the following presentation as a monoid:
  \begin{itemize}
  \item Generators: For each $\genr{} \in Γ$, the string diagram $N(γ) : 1 ⊗ ... ⊗ n → 1 ⊗ ... ⊗ n$ built from symmetries, followed by $\genr{}$ tensored with identities, followed by the inverse symmetry. See \Cref{fig:trace-generator} for an example.
  \item Equations: $N(α)\comp N(β) = N(β)\comp N(α) \iff \textsf{loc}(\genr[α]) \cap \textsf{loc}(\genr[β]) = \varnothing$, where $\comp$ denotes composition of string diagrams in diagrammatic (left-to-right) order.
  \end{itemize}
\end{lemma}
\begin{proof}
  We construct an isomorphism between the monoids. Let $s \in \traceEndo{Γ}$ be a string diagram. We can use interchange (\Cref{fig:interchange}) to impose a linear order of generators from left to right in the diagram, e.g. $\genrw[γ_1], ..., \genrw[γ_n]$. This is called putting $s$ in \emph{general position}, by perturbing generators at the same horizontal position \cite{joyal91}. We then split the string diagram into a sequence of slices, each containing one generator. For a slice with right (or left) boundary $\stk{k_1}{k_n}$, we can use the permutation $\stk{k_1}{k_n} \to \stk{1}{n}$ followed by its inverse (or vice-versa) to finally obtain $s$ as a sequence $N(γ_1)\comp ... \comp N(γ_n)$. Any other possible sequence of generators is obtainable by repeatedly interchanging generators: this is possible if and only if their locations are disjoint. Consequently, this defines a function from $\traceEndo{Γ}$ to the monoid presented above. Given that, as argued above, the slicing construction is unique up to interchanging independent generators, this defines a homomorphism. Conversely, given a generator $N(γ)$ in the presentation, we map this to the same string diagram in $\traceEndo{Γ}$. Again, it follows from interchange that this extends to a homomorphism, inverse to that above.
\end{proof}

\vspace*{-2mm} \begin{figure}[H]
  \centering
  \includegraphics[scale=0.75]{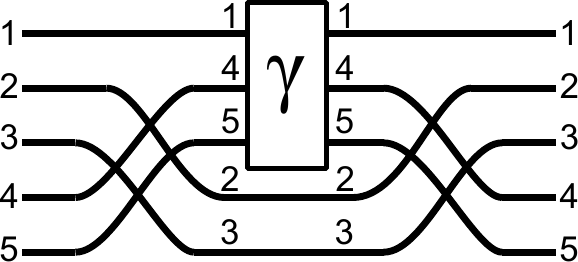}
  \caption{An example of a generator $N(γ)$ as in \Cref{prop:endo-pres}.} \label{fig:trace-generator}
\end{figure} \vspace*{-2mm}

We now show that \traceMonoids{} are isomorphic to the endomorphism monoids $\traceEndo{Γ}$.

\begin{lemma} \label{lemma:monoidaltrace}
  Let $I$ be an \independenceRelation{} on an alphabet $Σ$, and $Γ$ the \monoidalDistributedAlphabet{} induced by the corresponding \distribution{} (\Cref{prop:dists}). Then there is an isomorphism of monoids $\mathcal{T}(\Sigma, I) \cong \traceEndo{Γ}$.
\end{lemma}
\begin{proof}
  We use the presentation of the endomorphism monoid given in \Cref{prop:endo-pres}. Define a homomorphism $\alpha : \traceEndo{Γ} \to \mathcal{T}(Σ,I)$ by mapping generators $N(γ) ↦ [γ]$. Let $N(γ)\comp N(γ') = N(γ') \comp N(γ)$, then it follows $[γγ'] = [γ'γ]$ in $\mathcal{T}(Σ,I)$, since the former holds iff $\loc{γ} ∩ \loc{γ'} = \varnothing$, and so this extends to a homomorphism.
 Define a homomorphism $β : \mathcal{T}(\Sigma, I) → \traceEndo{Γ}$ by mapping generators $[γ] ↦ N(γ)$. $[γγ'] = [γ'γ$] holds iff $\loc{γ} ∩ \loc{γ'} = \varnothing$, iff $\loc{\genr[γ]} ∩ \loc{\genrw[γ']} = \varnothing$, iff $N(γ) \comp N(γ') = N(γ') \comp N(γ)$. Finally it is clear that $α$ and $β$ are inverses, and so witness an isomorphism of monoids. 
\end{proof}

The following theorem is now immediate: given a \monoidalTraceLanguage{} $L \subseteq \traceEndo{Γ}$ we obtain a trace language $L' \subseteq \mathcal{T}(Σ,I)$ using the isomorphism, and vice-versa:

\begin{theorem} \label{thm:monoidaltrace}
  \MonoidalTraceLanguages{} are exactly \MazurkiewiczTraceLanguages{}.
\end{theorem}

\Cref{lemma:monoidaltrace} also shows that composition of traces corresponds simply to concatenation of the corresponding monoidal traces. Diagrams like \Cref{fig:trace-example} are commonplace in the trace literature \cite{10.1007/978-3-642-32621-9_5,zielonka}. \Cref{thm:monoidaltrace} gives a formal basis for these diagrams as elements of \symmetricMonoidalLanguages{}.

\begin{remark} \label{rem:effectful}
  The idea of monoidal categories with a runtime is made precise by string diagrams for the \emph{effectful categories} of Román \cite{effectful22}. Free props over \monoidalDistributedAlphabets{}, considered as monoidal categories with \emph{multiple runtimes} suggest a further generalization of effectful categories, sketched as a setting for concurrency by Jeffrey \cite[Section 9.4]{jeffrey2}. We return to this in \Cref{sec:serial}, where effectful (premonoidal) categories will be used to equip a trace language with a new runtime that enforces a strict ordering of events.
\end{remark}

\section{Symmetric Monoidal Automata} \label{sec:automata}
\emph{Monoidal automata} give an alternative specification of the class of regular monoidal languages: they are analogues of finite-state automata in which transitions have multiple inputs and multiple outputs. Our paper \cite{earnshaw22} introduced monoidal automata for single-sorted, planar monoidal languages. However, the same data specifies an acceptor for single-sorted \symmetricMonoidalLanguages{}, if we inductively extend to \emph{props}, rather than planar monoidal categories.

In this section we introduce monoidal automata over \emph{multi-sorted} \monoidalGraphs{} and show how these recognize (multi-sorted) \symmetricMonoidalLanguages{}. In \Cref{sec:asynch}, we will see that the asynchronous automata of Zielonka \cite{zielonka} are a natural class of symmetric monoidal automata: those over \monoidalDistributedAlphabets{}.

\begin{definition} \label{defn:ndmonaut} \defining{linkndsmonaut}
  A non-deterministic monoidal semi-automaton is:
  \begin{itemize}
  \item an input alphabet, given by a finite \monoidalGraph{} $\Gamma$,
  \item an family of non-empty, finite state sets $\{Q_c\}_{c \in S_Γ}$ indexed by the sorts of $Γ$,
  \item for each $\gamma : c_1...c_n \to  c_1'...c_m'$ in $\Gamma$, a transition function $\Delta_\gamma : \prod_{i=0}^n Q_{c_i} \to \mathscr{P}(\prod_{j=0}^m Q_{c_j'})$.
  \end{itemize}
\end{definition}

As noted in \Cref{sec:languages}, there are several candidates for a notion of initial/final state. In the following, we take initial and final words $i, f$ over $\prod Q_c$. A monoidal semi-automaton equipped with initial and final words turns it into a (non-deterministic) \emph{monoidal automaton}.\defining{linkndmonaut}

For classical NFAs, the assignment $a \mapsto \Delta_a$ extends uniquely to a functor $\Sigma^{*} \to \textsf{Rel}$, the inductive extension of the transition structure from letters to words. We can similarly extend monoidal automata to string diagrams. First, we define the codomain prop, $\Rel[Γ][Q]$:

\begin{definition}
  For a family of sets $\{Q_c\}_{c \in S_Γ}$ indexed by the sorts of $Γ$ then $\Rel[Γ][Q]$ is the prop with:
  \begin{itemize}
  \item set of objects $S_Γ^{*}$,
  \item morphisms $c_1...c_n \to c_1'...c_m'$ functions $\prod_{i=1}^n Q_{c_i} \to \mathscr{P}(\prod_{j=1}^m Q_{c_j})$,
  \item composition is the usual composition of relations, i.e. $f ∘ g := μ ∘ \mathscr{P}(g) ∘ f$, where $μ$ is the canonical map from sets of subsets to subsets,
  \item $⊗$ is given on objects by concatenation,
  \item and on morphisms $f : \bigotimes_i {c_i} \to \bigotimes_j c'_j$ and $g : \bigotimes_k {d_k} \to \bigotimes_l d'_l$ by $f ⊗ g := \nabla \circ (f \times g)$, where $\nabla$ sends pairs of subsets to their cartesian product,
  \item symmetries $σ : c_1c_2 \to c_2c_1$ are functions $Q_{c_1} \times  Q_{c_2} \to \mathscr{P}(Q_{c_2} \times Q_{c_1}) : (q,q') ↦ \{(q', q)\}$.
  \end{itemize}
\end{definition}

Note that a non-deterministic \monoidalSemiAutomaton{} amounts to a morphism of \monoidalGraphs{} $\Gamma \to \mathscr{U}\Rel[Γ][Q]$. The adjunction $\Fprop{} \dashv \mathscr{U}$ implies that there is a unique extension to a strict monoidal functor $\Fprop{Γ} → \Rel[Γ][Q]$, which we call a non-deterministic \emph{symmetric} monoidal semi-automaton\defining{linksymmonaut}. This functor maps a string diagram to a relation. When this relation relates the initial word to the final word, the string diagram is \emph{accepted}:

\begin{definition}
  Let $\Delta : \Fprop{\Gamma} \to \Rel[Γ][Q]$ be a non-deterministic monoidal automaton with initial and final states $i,f \in S_Γ^{*}$. Then the symmetric monoidal language accepted by $\Delta$ is the set of morphisms $\Lang(\Delta) := \{ \alpha \in \Fprop{\Gamma} \mid f \in \Delta(\alpha)(i)\}$.
\end{definition}

Intuitively, a run of a symmetric monoidal automaton starts with a \emph{word} of states, whose subwords are modified by transitions corresponding to generators. Identity wires do not modify the states, and symmetries permute adjacent states.

\begin{observation}
  There is an evident correspondence between non-deterministic \monoidalAutomata{} and \regularMonoidalGrammars{}. The graphical representation of a grammar (such as \Cref{fig:circuit}) makes this most clear: it can also be thought of as the ``transition graph'' of a non-deterministic monoidal automaton. %
\end{observation}

\begin{remark} \label{rem:monad}
  We can further abstract our definition of monoidal automaton by noting that $\Rel[Γ][Q]$ is a sub-prop of the Kleisli category of the powerset monad $\mathscr{P}$, and that this monad could be replaced by another commutative monad \cite[Corollary 4.3]{power_robinson_1997}. For example, replacing $\mathscr{P}$ with the maybe monad, we obtain deterministic monoidal automata.%
\end{remark}

\section{Asynchronous Automata as Symmetric Monoidal Automata} \label{sec:asynch}
Asynchronous automata were introduced by Zielonka \cite{zielonka} as a true-concurrent operational model of recognizable trace languages, a well-behaved subclass of \traceLanguages{} analogous to regular languages. In this section we show they are precisely \symmetricMonoidalAutomata{} over \monoidalDistributedAlphabets{}, which leads to the following theorem: %

\begin{theorem} \label{thm:rec}
  Recognizable trace languages are exactly \regularSymmetricMonoidalLanguages{} over \monoidalDistributedAlphabets{}.
\end{theorem}

We recall the definition of asynchronous automata, before turning to \monoidalAutomata{}.

\begin{definition}[Asynchronous automaton \cite{zielonka}] \label{defn:asyncaut}
  Let $(\Sigma_1,..., \Sigma_k)$ be a \distribution{} of an alphabet $\Sigma$. For each $1 \leqslant i \leqslant k$, let $Q_i$ be a non-empty finite set of states, and for each $\sigma \in \Sigma$ take a transition relation $\Delta_\sigma : \prod_{i \in \loc{σ}} Q_i \to \mathscr{P}(\prod_{i \in \loc{σ}} Q_i)$. This defines a global transition relation on the set $Q := \prod_{i=1}^k Q_i$ as follows:

$(q_1, ..., q_k) \xrightarrow[]{\sigma} (q_1',...q_k') \iff q_i = q_i'$ for $i \notin \loc{σ}$ and $(q_{i_1}', ... , q_{i_j}') \in \Delta_{\sigma}(q_{i_1}, ..., q_{i_j})$ where $\{i_1, ..., i_j\} \in \loc{σ}$. Finally let $\overrightarrow{i} \in Q, F \subseteq Q$ be initial and final words of states.
\end{definition}

The global transition relation for $σ$ leaves unchanged those states at locations in the complement of $\loc{σ}$, and otherwise acts according to the local transition $Δ_σ$. An asynchronous automaton has a language over $\Sigma$ given by the extension of the transition relation to words. Moreover, asynchronous automata have a language of \MazurkiewiczTraces{} over the distribution of $Σ$: a trace in $\mathcal{T}(Σ,I)$ is accepted when all of its serializations are accepted, which happens when one of its serializations is accepted \cite[p. 109]{zielonka}. \emph{Recognizable trace languages} are defined algebraically as those whose syntactic congruence is of finite index \cite{zielonka}. Zielonka's theorem says that they also have an operational characterization:

\begin{theorem}[Zielonka \cite{zielonka}] \label{thm:zielonka}
  Asynchronous automata accept precisely the \emph{recognizable trace languages}.
\end{theorem}

\Cref{defn:asyncaut} closely resembles that of \symmetricMonoidalAutomata{}. Indeed, asynchronous automata are precisely \symmetricMonoidalAutomata{} over \monoidalDistributedAlphabets{}:

\begin{proposition} \label{prop:monasync}
  For an asynchronous automaton $𝒜$, there is a \symmetricMonoidalAutomaton{} over a \monoidalDistributedAlphabet{} with the same trace language, and vice-versa.
\end{proposition}
\begin{proof}
  An asynchronous automaton with multiple final state words can be normalized to a single final state word in the usual way by introducing a new final state word and modifying transitions appropriately. Then a \symmetricMonoidalAutomaton{} can be constructed by taking the \monoidalDistributedAlphabet{} associated to the distribution of $Σ$ (\Cref{prop:dists}), the same transition relations, initial and final state words. We show that the languages coincide. Let $w \in \Lang(𝒜)$, and consider the corresponding trace $[w]$. Using \Cref{lemma:monoidaltrace}, we can produce the corresponding monoidal trace. By construction, this is accepted by the \symmetricMonoidalAutomaton{} defined above. The converse is analogous.
\end{proof}\vspace*{-2mm}

As a corollary, we can invoke \Cref{thm:zielonka} to obtain \Cref{thm:rec}. In contrast to asynchronous automata, the constructed \symmetricMonoidalAutomaton{} directly accepts traces qua string diagrams, rather than a language of words corresponding to a trace language. %

\begin{observation} \label{obs:async-are-monoidal}
  Jesi, Pighizzini, and Sabadini \cite{Jesi1996} introduced probabilistic asynchronous automata. Initial and final states, and transition relations are replaced by initial and final distributions, and stochastic transitions. These are precisely what are obtained if the powerset monad in our definition of non-deterministic \monoidalAutomaton{} (\Cref{rem:monad}) is replaced with the \emph{distribution monad} \cite{nlab:distmonad}, whose Kleisli category has morphisms stochastic matrices.
\end{observation}

\section{Serialization via Premonoidal Categories} \label{sec:serial}

Trace theorists often consider trace languages to be word languages with the property of \emph{trace-closure} with respect to an \independenceRelation{} \cite{maarand2019reordering}: if $u \in L$ and $u \equiv_{I} v$ then $v \in L$. These languages arise as preimages of \traceLanguages{} along the quotient map $q_{Σ,I} : Σ^{*} → \mathcal{T}(Σ,I)$. For $L \subseteq \mathcal{T}(Σ,I)$ a \traceLanguage{}, $q_{Σ,I}^{-1}(L) \subseteq Σ^{*}$ is its \emph{flattening} or \emph{serialization}.

In this section we show that the serialization of \monoidalTraceLanguages{} can be carried out using the algebra and string diagrams of symmetric premonoidal categories. Premonoidal categories are like monoidal categories, except interchange (\Cref{fig:interchange}) does not hold in general. The free (symmetric) premonoidal category on a \monoidalGraph{} was described using \stringDiagrams{} by Román \cite{effectful22}. The idea is simple: the \stringDiagrams{} are the same as for props, but an extra string (the ``runtime'') threads through each generator, preventing interchange. \Cref{fig:premonoidal} shows two premonoidal morphisms ${\color{mylightblue} \bullet} ⊗ {\color{mygreen} \bullet} → {\color{mylightblue} \bullet} ⊗ {\color{mygreen} \bullet}$ that are not equal:

\begin{figure}[H]
  \centering
  \includegraphics[scale=0.8]{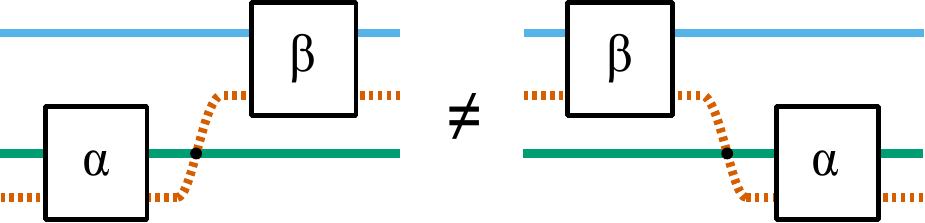}
  \caption{In the free premonoidal category over a \monoidalGraph{}, generators are augmented by a string on a new object called the \emph{runtime} (dashed red). This prevents interchange (cf. \Cref{fig:interchange}).}
  \label{fig:premonoidal}
\end{figure}

In \Cref{app:premonoidal}, we explain in more detail the construction of the free symmetric premonoidal category $\Fprem{Γ}$ on a \monoidalGraph{} $Γ$ using string diagrams. In particular, the runtime string appears only once in each string diagram, reflecting that premonoidal categories do not have a tensor product on morphisms. The endomorphism monoid $\preEndo{Γ}$ is now the free monoid over the boxes of $Γ$, since the runtime prevents interchange:

\begin{proposition} \label{prop:pre-endo}
  Let $Γ$ be a \monoidalDistributedAlphabet{}. Then $\preEndo{Γ} ≅ B_Γ^{*}$, where $B_Γ$ is the set of boxes of $Γ$.
\end{proposition}
\begin{proof}\emph{(Sketch)}
  By augmenting the generators of $Γ$ with a new runtime, we create a \monoidalDistributedAlphabet{} in which every generator depends on every other, that is, the independence relation is empty. Thus the corresponding \traceMonoid{} is simply $B_Γ^{*}$. From here, we can follow the idea of \Cref{lemma:monoidaltrace}.
\end{proof}

We can define a morphism of monoids $\mathfrak{q}_Γ : \preEndo{Γ} → \traceEndo{Γ}$ by presenting $\preEndo{Γ}$ as in \Cref{prop:endo-pres}, and defining $\mathfrak{q}_Γ$ on generators by erasing the runtime string. \Cref{thm:preimage} then follows immediately from the definitions along with \Cref{lemma:monoidaltrace} and \Cref{prop:pre-endo}:

\begin{theorem} \label{thm:preimage}
  For every alphabet $B_Γ$, the following square of monoid homomorphisms commutes, where $q$ is the quotient monoid homomorphism.
\vspace*{-2mm}\[\begin{tikzcd}[ampersand replacement=\&]
	{B_Γ^{*}} \& {\tau(B_Γ,I)} \\
	{\preEndo{Γ}} \& {\traceEndo{Γ}}
	\arrow["\cong"', from=1-1, to=2-1]
	\arrow["q", from=1-1, to=1-2]
	\arrow["{≅}", from=1-2, to=2-2]
	\arrow["{\mathfrak{q}_Γ}", from=2-1, to=2-2]
\end{tikzcd}\]
\end{theorem}

As a result, the preimage of a \monoidalTraceLanguage{} under the morphism $\mathfrak{q}_Γ$ corresponds to the serialization of that language.

\section{Conclusion} \label{sec:conclusion}

There are several directions in which our theory could be developed. A semi-independence relation drops symmetry from an independence relation: it is simply an irreflexive relation.  This gives rise to the theory of semicommutations \cite{semicommutations}, in which \emph{directed} commutations may occur e.g. $ab \to ba$, but not vice-versa. This allows for a more fine-grained specification of concurrency. In terms of monoidal languages, it suggests consideration of monoidal distributed alphabets in which the sources and targets of generators may differ.

As noted in \Cref{rem:effectful}, our treatment of trace languages suggests a generalization of the notion of \emph{effectful category} \cite{effectful22} (which include premonoidal categories), in which there are \emph{multiple runtimes}. This would enable a semantics for concurrent systems in which we can consider not only \emph{atomic} actions, but also actions with input and output types. We plan to pursue this axiomatically in future work.

Mazurkiewicz originally introduced traces to give semantics to Petri nets, and showed that this semantics is compositional with respect to \emph{synchronization} of traces \cite{mazurkiewicz}. Petri nets have been given semantics in monoidal categories \cite{baez2021categories,MESEGUER1990105}, and so the precise relationship of our monoidal formulation of traces to Petri nets remains to be worked out. In particular, this would involve understanding trace synchronization in terms of monoidal categories.

Finally, proofs of Zielonka's theorem (\Cref{thm:zielonka}, see \cite{zielonka} for details) remain highly technical, despite several simplifications since Zielonka's version. Investigation of whether the algebra of monoidal categories might yield further simplifications is an intriguing direction.

\bibliography{main}

\appendix

\section{Symmetric Strict Premonoidal Categories and Functors} \label{app:premonoidal}
We recall the definitions of (symmetric) strict premonoidal categories and their functors. For more details, see the papers \cite{power_robinson_1997,effectful22}.

\begin{definition}
  A strict premonoidal category is a category $𝒞$ equipped with:
  \begin{itemize}
  \item for each pair of objects $A,B \in 𝒞$ an object $A⊗B$,
  \item for each object $A \in 𝒞$ a functor $A ◁ -$ whose action on objects sends $B$ to $A ⊗ B$,
  \item for each object $A \in 𝒞$ a functor $- ▷ A$ whose action on objects sends $B$ to $B ⊗ A$,
  \item a unit object $I$,
  \end{itemize}
  such that,
  \begin{itemize}
  \item for each object $A \in 𝒞$, strict unitality $I ⊗ A = A = A ⊗ I$ holds,
  \item for each triple of objects $A,B,C \in 𝒞$, strict associativity $A ⊗ (B ⊗ C) = (A ⊗ B) ⊗ C$ holds.
  \end{itemize}
\end{definition}

The families of functors $A ◁ -, - ▷ A$ are called the \emph{whiskerings} with $A$: in a premonoidal category we do not have a tensor product of morphisms in general, but we can put an identity on either side of a morphism. A morphism $f : A → B \in 𝒞$ is \emph{central} if for every morphism $g : C → D$, $(B ◁ g)∘(f ▷ C) = (f  ▷ C)∘(A ◁ g)$, in other words, $f$ is central if it interchanges with every other morphism $g$.

\begin{definition}
A strict premonoidal category is symmetric if it is further equipped with a natural isomorphism whose components $c_{A,B} : A ⊗ B → B ⊗ A$ are central and such that $c_{B,A} ∘ c_{A,B} = 1_{A ⊗ B}$.
\end{definition}

\begin{definition}
  A strict premonoidal functor $F : 𝒞 → 𝒟$ is a functor sending central morphisms to central morphisms and such that $F(I_𝒞) = I_𝒟, F(X ⊗ Y) = F(X) ⊗ F(Y)$.
\end{definition}

\subsection{String Diagrams for Premonoidal Categories} 
We recall the construction of the free symmetric strict premonoidal category over a \monoidalGraph{}. This is a special case of the construction of free \emph{effectful categories} in \cite[Section 2.3]{effectful22}.

We first define the \emph{runtime monoidal graph} over a \monoidalGraph{}, which augments the generators with a new wire:

\begin{definition}
  Let $𝒢$ be a \monoidalGraph{}. Let $R$ be a sort disjoint from $S_𝒢$. The \emph{runtime \monoidalGraph{}} $𝒢_R$ has sorts $S_𝒢 + \{R\}$ and for each generator $γ : S_1...S_n \to S'_1...S'_m$ in $𝒢$ a generator $γ : RS_1...S_n \to RS'_1...S'_m$.
\end{definition}

Graphically we can depict $𝒢_R$ as in \Cref{fig:runtime-graph} (right):

\begin{figure}[H]
  \centering
  \includegraphics[scale=1.5]{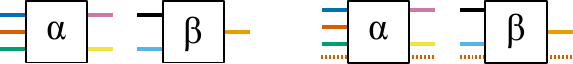}
  \caption{Left: A monoidal graph $𝒢$. Right: the associated runtime monoidal graph $𝒢_R$, where the new sort $R$ is drawn as a dashed string.} \label{fig:runtime-graph}
\end{figure}

\begin{definition}
  The symmetric \emph{runtime monoidal category} is the free prop $\Fprop{𝒢_R}$ on $𝒢_R$.
\end{definition}

\begin{theorem}
  The free symmetric strict premonoidal category $\Fprem{𝒢}$ on a \monoidalGraph{} $𝒢$ has set of objects $S_𝒢$ and a morphism $S_1 ⊗ ... ⊗ S_n → S_1' ⊗ ... ⊗ S_m'$ is a morphism $R ⊗ S_1 ⊗ ... ⊗ S_n → R ⊗ S_1' ⊗ ... ⊗ S_m'$ in the symmetric runtime monoidal category.
\end{theorem}
\begin{proof}
  The proof follows \cite[Theorem 2.14]{effectful22}, in the case where $\mathscr{V}$ is empty, and taking instead the free symmetric strict monoidal category.
\end{proof}

In particular note that we no longer have a tensor product of morphisms in $\Fprem{G}$, since the runtime must appear only once in each domain and codomain, but we do have whiskerings for each object.

Consequently the string diagrams for morphisms $A → B$ in $\Fprem{𝒢}$ are just morphisms $R ⊗ A → R ⊗ B$ in the symmetric runtime monoidal category \cite[Corollary 2.15]{effectful22}.

\end{document}